\documentclass[Journal]{IEEEtran}
\usepackage{amsmath}
\usepackage{amsthm}
\usepackage{algorithm}
\usepackage{algpseudocode}
\usepackage{graphics}
\usepackage{epsfig}
\usepackage{mathrsfs}
\usepackage{amssymb}
\usepackage[OT1]{fontenc}

\begin{document}
\newtheorem{emp}{Example}
\newtheorem{cn}{Conjecture}
\newtheorem{lm}{Lemma}
\newtheorem{thm}{Theorem}
\newtheorem{cor}{Corollary}
\newtheorem{df}{Definition}
\newtheorem{pf}{Proof}
\newtheorem{alg}{Algorithm}
\newtheorem{rmk}{\bf Remark}


\title{Robustness in Chinese Remainder Theorem}
\author{ Hanshen~Xiao,
            Yufeng~Huang,
            Yu Ye
            and~Guoqiang~Xiao
  \thanks{Hanshen Xiao is with MIT Computer Science and Artificial Intelligence Laboratory, Massachusetts Institute of Technology, Cambridge, USA (email: hsxiao@mit.edu),

  Yufeng Huang is with the Department of Electronic Engineering, Tsinghua University, Beijing, China (email: huang-yf15@mails.tsinghua.edu.cn)

  Yu Ye is with the School of Electrical Engineering, Royal Institute of Technology (KTH), Stockholm, Sweden (email: yu9@kth.se)

  Guoqiang Xiao is with the College of Computer and Information Science, Southwest University, Chongqing, China (email: gqxiao@swu.edu.cn)}}


\let\lc\langle
\let\rc\rangle

\date{}

\maketitle
\bibliographystyle{apalike}
\begin{abstract}
Chinese Remainder Theorem (CRT) has been widely studied with its applications in frequency estimation, phase unwrapping, coding theory and distributed data storage. Since traditional CRT is greatly sensitive to the errors in residues due to noises, the problem of robustly reconstructing integers via the erroneous residues has been intensively studied in the literature. In order to robustly reconstruct integers, there are two kinds of traditional methods: the one is to introduce common divisors in the moduli and the other is to directly decrease the dynamic range. In this paper, we take further insight into the geometry property of the linear space associated with CRT. Echoing both ways to introduce redundancy, we propose a pseudo metric to analyze the trade-off between the error bound and the dynamic range for robust CRT in general. Furthermore, we present the first robust CRT for multiple numbers to solve the problem of the CRT-based undersampling frequency estimation in general cases. Based on symmetric polynomials, we proved that in most cases, the problem can be solved in polynomial time efficiently. The work in this paper is towards a complete theoretical solution to the open problem over 20 years.

\end{abstract}


\begin{section}{Introduction}

Conventional Chinese Remainder Theorem (CRT) is to reconstruct a single integer by its residues modulo several moduli. There exist many cases, which require a large ring to be represented by several small distributed subrings, in engineering applications. Therefore, CRT has been widely studied and applied in error correcting codes \cite{2000chinese}, \cite{ETRI2016}, \cite{soft2000}, frequency estimation from undersampled waveforms \cite{1997multiple}, \cite{1997}, \cite{2000}, phase unwrapping \cite{2007phase}, \cite{2008phase}, wireless sensor networks \cite{2012storage}, modular multiplication computing \cite{access2016}, \cite{ICASSP2017} etc. However, conventional CRT aims to reconstruct a single integer, which also requires that moduli are co-prime and the residues are error-free. A small error in residues may lead to a large error in reconstruction. Therefore, derived from undersampling frequency estimation \cite{1997}, \cite{1997multiple}, generalized CRT (GCRT) for multiple integers and robust CRT (RCRT) receive considerable attention. During last two decades, a great number of researches promote the development of CRT-based signal parameter estimation and related application fields, which can be concluded as the following three stages.

 i) \textbf{Generalized CRT (GCRT) for multiple integers}. The problem is brought in \cite{1997} in 1997 and a case of narrow-band signals has been studied. It can be mathematically stated as follows. Given $L$ moduli  $\{m_l, l=1,2,...,L \}$, which stand for the sampling frequencies in the model, and $N$ integers $\{X_i, i=1,2,...,N\}$, standing for the frequencies to be estimated, are required to be determined from $L$ residue sets $R_l=\{r_{il} | r_{il}=\langle X_i \rangle_{m_l}, l=1,2,...,L \}$, where $\langle X_i \rangle_{m_l}$ denotes the residue of $X$ modulo $m_l$.  Here $\{r_{il}\}$ are assumed to be error-free while not ordered. The difficulty of reconstruction is due to the missing correspondence between residues and integers in each residue set. The algorithm proposed in \cite{2007sharpened} solves the problem in general with a sharpened lower bound  of $lcm(m_1, m_2, ..., m_l)$, where $lcm$ denotes the least common multiple. Xiao et al. \cite{2015newcondition} proposed an algorithm to deal with a narrow-band case, where $\min_{l} m_l > \max_{i,j} |X_i-X_j|$, in polynomial time complexity. The first general GCRT scheme with polynomial running time is proposed in \cite{2016symmetric} with a loosen restriction of moduli, which introduces symmetric polynomials to compensate the lack of corresponding relationship.

ii) \textbf{Robust CRT (RCRT) for a single integer}. To overcome sensibility to error, redundancy is introduced by adding a common factor $\Gamma$ in moduli. For a group of coprime integers $\{M_1,M_2,...,M_N\}$, a common factor $\Gamma$ is introduced and the moduli $\{m_i\}$ are in the form $\{m_i= \Gamma M_1, m_2= \Gamma M_2, ... , m_N= \Gamma M_N \}$. With the redundancy, the key idea behind almost all RCRT schemes is to correctly recover the folding number. The search-based RCRT are proposed in \cite{search1}, \cite{search2}, \cite{search3}. In 2010, the first closed-form RCRT is developed in \cite{2010closed}. The main results of above-mentioned work can be concluded as follows. If the error in each residue is upper bounded by $\delta$, which is less than a quarter of $\Gamma$, an integer can be robustly recovered, where the reconstruction error is also upper bounded by $\delta$.  Thus reconstruction with RCRT may not be precise but robust in the sense of small error in each residue and also small deviation in reconstruction. A variant RCRT of \cite{2010closed} with a simpler form is developed in \cite{2014arxiv}. Following those, a recursive reconstruction, called multistage RCRT, is proposed in \cite{2014multistage} for more general moduli. In multistage RCRT, a part of residues are combined each time to robustly recover the residue modulo the $lcm$ of corresponding moduli of residues selected. Thus it is possible to achieve robustness even if not all the moduli share a common factor. Besides, a lower bound of the product of all the moduli under a given $\delta$ is developed and proved in \cite{2017notes}, which also derives the optimal modulus selection scheme. In \cite{2015mle}, a fast maximum likelihood estimation (MLE) based RCRT is proposed, where the complexity is reduced to searching within $L$ elements. Furthermore, the research on RCRT also leads to the analysis of Lee-metric based remainder code, which regards representing an integer with its residues as encoding. The error correcting of the Lee-metric based remainder code is also based on RCRT, where a generalization over the polynomial ring has recently been studied in \cite{poly}. A lattice-based decoding algorithm is proposed in \cite{2004lattice}. As already mentioned, such code is different from the classic remainder code \cite{2000chinese}, \cite{ETRI2016}, where error is measured with the Hamming weight and redundancy is introduced via redundant moduli. To guarantee the robustness, a natural question raised is about the relationship between the minimal code distance (corresponding to $\delta$) and the dynamic range under given moduli, which also determines the information ratio. In \cite{2017towards}, a closed-form relationship has been developed under the case of $L=2$, i.e, two moduli. However, in order to apply it in general cases, more advanced techniques are required.

 iii) \textbf{Robust CRT for multiple integers}. This is the last stage to provide a complete theoretical solution for undersampling signal parameter estimation with CRT-based methods. Up to now, there are few works concerning the issue of robustly reconstructing multiple integers with erroneous residues. The RCRT for two integers is recently studied in \cite{2016two}. A scheme for the case of narrow-band signals is proposed in \cite{2017notes}. However, there is no result concerning the general case. The difficulty lies  both in absence of correspondences and errors caused by noise.

In this paper, we further the research in all the three above-mentioned aspects. Our main contributions can be concluded as follows.
 \begin{itemize}
 \item We develop the first nearly closed-form generalized RCRT for multiple integers, called robust generalized CRT (RGCRT), as a complete solution to CRT-based frequency estimation from undersampling waveforms. We prove that the issue of GRCRT can be simplified to the issue of GCRT and develop a new class of symmetric polynomials to exponentially sharpen the lower bound of moduli in \cite{2016symmetric}.
 \item Second, we take further insight into the general relationship between the error bond $\delta$ and dynamic range $K$ of the integer to be recovered. A novel pseudo metric, called shift pseudo metric, has been proposed to measure the distance between residue vectors and also disclose the essence relation of $\delta$ and $K$ from a geometry perspective. To show its potential application in coding theory, we develop a search-based error correcting scheme and provide probability of success to find out satisfied moduli for the code construction. It is a general framework to study such remainder code.
 \end{itemize}

The rest contents are organized as follows. The background is given in Section II. In Section III, we introduce the shift pseudo metric and present the results concerning the relationship between $\delta$ and $K$ for general moduli. In section IV, we give a shorter proof of the results obtained in \cite{2017towards} under the case of two moduli. A search-based scheme and probability analysis of random modulus selection are proposed and analyzed in Section V. One of our main contributions to develop the GRCRT is shown in Section VI.  
The conclusion of the paper follows in Section VII.

\end{section}

\begin{section}{Background}

\noindent We begin with the problem introduced in \cite{1997} that multiple frequency determination from multiple undersampled waveforms. $N$ frequencies, $\{X_1,X_2,...,X_N\}$, in $Hz$ assumed to be integers need to be determined in a super positioned waveform $x(t)$, i.e.,\\
\begin{equation}
x(t)=\sum_{i=1}^{N}A_ie^{2\pi jX_it}
\end{equation}
where $A_i$ are nonzero complex coefficients for $1\leq i\leq N$. Let $\{m_1,m_2,...,m_L\}$, arranged in an ascending order, be the sampling frequencies which may be much less than the frequencies, $\{X_i\}$. Then the undersampled waveforms become\\
\begin{equation}
x_{m_l}[n]=\sum_{i=1}^{N}A_ie^{\frac{2\pi jX_in}{m_l}},~n\in \mathbb{Z}
\end{equation}
under a sampling frequency $f_s=m_l$ Hz, where $\mathbb{Z}$ denotes the set of integers. Applying the $m_l$-point Discrete Fourier Transform (DFT) for $x_{m_l} [n]$, we have\\
$$
DFT_{m_l}(x_{m_l}[n]) [k]=\sum_{i=1}^{N}A_i \mathbf{1}(k-\langle X_i \rangle _{m_l})
$$
where $\mathbf{1}(k-\langle X_i \rangle _{m_l})=1$ when $k=\langle X_i \rangle _{m_l}$, otherwise it turns out to be 0. Here $\langle X_i \rangle _{m_l}$ denotes the residue of $X_i$ modulo $m_l$,  i.e., $X_i - m_l\lfloor X_i/m_l \rfloor$. Thus the problem is transformed to recover $X_i$ from $L$ residue sets $R_l = \{ r_{il} = \langle X_i \rangle _{m_l} | i=1,2,...,N \}$, $l=1,2,...,L$, without knowing the correspondence between $X_i$ and $r_{il}$.

   When we consider the noise introduced, for a complex-valued waveform,
\begin{equation}
x(t)=\sum_{i=1}^{N}A_ie^{2\pi jX_it} + \omega(t)
\end{equation}
where $\omega(t)$ is additive noise. Thus
$$
DFT_{m_l}(x_{m_l}[n]) [k] =\sum_{i=1}^{N}A_i \mathbf{1}(k-\langle X_i \rangle _{m_l}) +W_l [k]
$$
Based on \cite{2013noise},\cite{2015mle}, the variance of $W_l [k]$ is proportional to a fraction of $m_l$. A robust reconstruction of $X_i$ is expected to tolerate small errors $\Delta r_{il}$ introduced in each residue $r_{il}$, i.e., to robustly construct $X_i$ with the erroneous $\widetilde{r}_{il}$. The robustness is defined as that for a prescribed error bound $\delta$, if $\max_{i,l} |\Delta {r}_{il}| < \delta$, then the estimation $\widetilde{X}_{i}$ of $X_{i}$ satisfies that $\max_{i} |\widetilde{X}_{i}-X_{i}| < \delta$.

    The first closed-form RCRT for the case that $N=1$ is proposed in \cite{2010closed} and its variant form is presented in \cite{2014arxiv}. The key idea in \cite{2010closed} and \cite{2014arxiv} is to recover the folding number correctly, which is similar to the previous searching based methods. Follow-up research such as maximum likelihood estimation (MLE) based \cite{2015mle} or multi-stage based \cite{2014multistage} RCRT is derived from it. Define the set of selected moduli as $\mathscr{M} = \lbrace m_l | l = 1, 2, 3\dots, L \rbrace$. Let $\text{lcm}(\mathscr{M})$ denote the least common multiple (lcm) of all moduli, i.e., $\text{lcm}(\mathscr{M})= \text{lcm} (m_1,m_2,...,m_L)$ and $\mathcal{M}(K)$ denote the set of all multiples of a modulus, i.e., $\lbrace Y~ |~0\leq Y \leq K, \exists m_l \in \mathscr{M}, m_l | Y \rbrace$. 
    When there is only one integer $X$ to be recovered,  $X$ can be represented as
    $$X=k_lm_l+ \widetilde{r}_l - \Delta{r_l}=k_lm_l+r_l$$
     for $l=1,2,...,L$. Here, $\widetilde{r_l}$ denotes the erroneous residue of $X$ modulo $m_l$ and $\Delta{r_l}$ is the error introduced correspondingly. Therefore, $k_1m_1-k_lm_l=\widetilde{r}_l-\widetilde{r}_1-(\Delta{r_l}-\Delta{r_1})$, which can be transformed to
    \begin{equation}
    \begin{aligned}
    &k_1\frac{m_1}{gcd(m_1,m_l)}-k_l\frac{m_l}{gcd(m_1,m_l)} \equiv k_1\frac{m_1}{gcd(m_1,m_l)}\\
    &\equiv \frac{\widetilde{r}_l-\widetilde{r}_1}{gcd(m_1,m_l)}+\frac{\Delta{r_1}-\Delta{r_l}}{gcd(m_1,m_l)} \mod \frac{m_l}{gcd(m_1,m_l)}
    \end{aligned}
    \end{equation}
   When $| \frac{\Delta{r_1}-\Delta{r_l}}{gcd(m_1,m_l)}|<\frac{1}{2}$, it is not hard to obverse that
   $$[\frac{\widetilde{r}_l-\widetilde{r}_1}{gcd(m_1,m_l)}] = \frac{r_l-r_1}{gcd(m_1,m_l)}$$
  Thus the residue of $k_1$ modulo $\frac{m_l}{gcd(m_1,m_l)}$ has been found. Here $[\epsilon]$, $\epsilon \in \mathbb{R}$, is the round operation, i.e., $[\epsilon]$ is the integer closest to $\epsilon$. Therefore $k_1$ can be recovered with conventional CRT correctly. It is obvious that
   $$|k_1m_1+[\frac{\sum_{l=1}^{L} \widetilde{r}_l}{L}]-X| \leq \max_{l} |\Delta{r_l}|,$$
  which achieves error control. Now we view the RCRT from a geometry perspective. Basically speaking, what we are interested in is that under the circumstance where each component (residue) may have an error upper bounded by $\delta$, for an integer within $[0,K]$ where $K$ is the dynamic range and $K<\text{lcm}(\mathscr{M})$, how to robustly reconstruct the integer with a deviation bounded by $\delta$ as well. An $L$-dimensional residue vector of an integer $X$ is constructed by its residues modulo $L$ moduli, which can be regarded as a point of $L$-dimensional Euclid space restricted in a hypercube, whose edges are determined by moduli, called residue space. Integers within $[0,K]$ can be considered as located on a family of parallel lines with direction vector $(1,1,...,1)$. For a more intuitive sense, we give two examples of residue presentation in the cases of 2-modulus and 3-modulus in Fig. 1 and Fig. 2, respectively.
  \begin{figure}[htbp]
  \centering
  \includegraphics[width=0.45\textwidth]{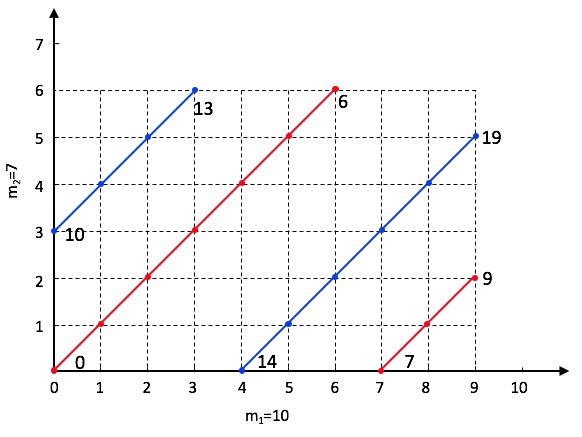}
  \caption{2-D Case with Moduli $m_1=10$ and $m_2=7$}
  \label{fig:digit}
\end{figure}

\begin{figure}[htbp]
  \centering
  \includegraphics[width=0.45\textwidth]{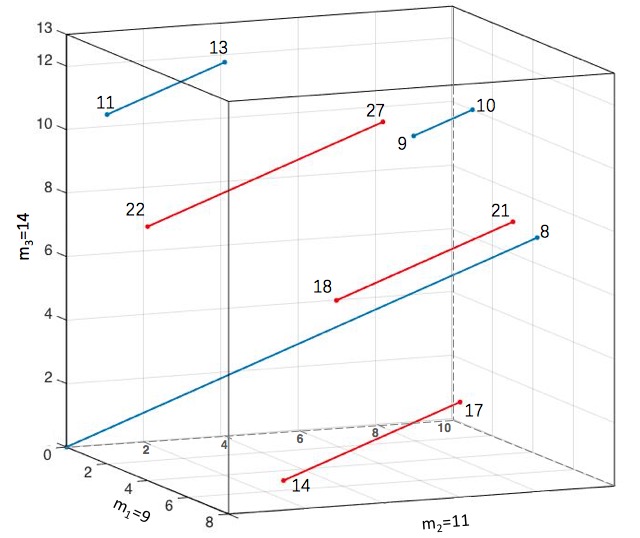}
  \caption{3-D Case with Moduli $m_1=9$, $m_2=11$ and $m_3=14$}
  \label{fig:digit}
\end{figure}

  The lines are denser as $K$ increases. $\mathcal{M}(K)$ are consisted of all starting points of the parallel lines within the hypercube. We will see later in Corollary 2 that robustly recovering $X$ is to find out in which line the residue vector of $X$ is located. Therefore intuitively the robustness can be achieved if two hypercubes, whose centers are on two different lines respectively with the same length of edges as $2\delta$, do not share intersection.  Therefore, there are two research problems.
  \begin{itemize}
  \item(1) Given a dynamic range $K$ and a $\mathscr{M}$, what is the minimal distance between the parallel lines, in which the residue vectors of $X\in[0,K]$ is located.
  \item(2) Given a dynamic range $K$ and an error bound $\delta$, from which a lower bound of minimal distance between the parallel lines can be derived, how to select a set of moduli.
  \end{itemize}
  Before we start, the notations used in this paper are listed in Table I for clarity.

  \begin{table}
\caption{List of Notations}
\begin{tabular*}{8.8cm}{lll}
\hline
Notations &~~~~~~~~~~~~~~~~ Explanation    \\
\hline
$L$  & The number of moduli selected  \\
$ m_l, l=1,2,...,L$ & Moduli selected  \\
$\mathscr{M}$ & $\mathscr{M}=\{m_l, l=1,2,...,L\}$ \\
$N$ & The number of integers to be recovered\\
$X_i, i=1,2,...,N$ & Integers to be recovered\\
$\{\widetilde{X}_{i}\}$ & Estimation of $X_i$ \\
$K$ & The dynamic range of $X_i$ \\
$\mathcal{M}(K)$ & Multiples of any modulus within $[0,K]$\\
$\delta $ & Error bound for the residue in each modulus\\
$\delta(K) $ & Maximal error bound for the given $K$ \\
$r_{il}$ & The residue of $X_i$ modulus $m_l$ \\
$\widetilde{r}_{il}$ & The erroneous residue of $X_i$ modulus $m_l$ received \\
$\Delta r_{il}$ & Error introduced in $\widetilde{r}_{il}$ comparing to $r_{il}$ \\
$\rho$ & Shift pseudo metric developed in the paper  \\
$\mathbf{x}$ & The residue vector of an integer $X$\\
$\mathbf{\tilde x}$ & The erroneous residue vector received \\
$B$ & $\max_{1 \leq i,j \leq N} | X_i - X_j |$ \\
$X_0$ & The largest element in $\mathcal{M}(K)$ such that $X_0 \leq X$ \\
$r^{c}_i$ & The common residue, i.e., $\langle X_i \rangle_{\Gamma}$\\

\hline
\end{tabular*}
\end{table}

   \end{section}

\begin{section}{Shift Pseudo-Metric and Error Bound}
\noindent For an integer $X$, let $\bf{x}$ be its residue vector, i.e, $\mathbf{x}=(\langle X \rangle_{m_1},\langle X \rangle_{m_2},\dots,\langle X \rangle_{m_L})$. In order to formalize the robustness and measure the distance between residue vectors, we propose a new pseudo metric, called $Shift~Pseudo~Metric$, with the definition as follows.

   \subsection{Shift Pseudo Metric}
     \begin{df}[Shift Pseudo Metric]\label{sm} For two residue vectors $\mathbf{x, y}$, where $X, Y\in [0, \text{lcm}(\mathscr{M})-1]$, \footnote{The axioms of a pseudo metric $\rho$ is slightly different from that of a metric, where there may exist two  elements $\mathbf{x}$ and $\mathbf{y}$, such that $\rho(\mathbf{x},\mathbf{y})=0$ whereas $\mathbf{x} \not = \mathbf{y}$.}
   $$
        \rho(\mathbf{x}, \mathbf{y}) = \max_{1\le l<j\le n} |(\lc X\rc_{m_l} - \lc X\rc_{m_j}) - (\lc Y \rc_{m_l} - \lc Y\rc_{m_j})|
  $$
    \end{df}
     We verify that this is a pseudo metric.\\
     \indent (1) $\rho(\mathbf{x}, \mathbf{x}) = 0$ is obvious. \\
       \indent (2) $\rho(\mathbf{y}, \mathbf{x}) =  \rho(\mathbf{x}, \mathbf{y})$ is obvious from the symmetry of $X,Y$ in the definition.\\
      \indent  (3) Note that for any two sequences $\{a_i\}$ and $\{b_i\}$ with a same length,
        \begin{equation}\label{def to sub}
            \max_i\lbrace a_i \rbrace + \max_i \lbrace b_i\rbrace \ge \max_i\lbrace a_i + b_i\rbrace
        \end{equation}
        Substitute  $|(\lc Y\rc_{m_l} - \lc Y\rc_{m_j})-(\lc X\rc_{m_l} - \lc X\rc_{m_j})|$, $|(\lc X\rc_{m_l} - \lc X\rc_{m_j})-(\lc Z\rc_{m_l} - \lc Z\rc_{m_j})|$ for $a_i$, $b_i$ in (\ref{def to sub}) respectively, and we have
        \begin{equation}
            \rho(\mathbf{y}, \mathbf{x}) + \rho(\mathbf{x}, \mathbf{z}) \ge \rho(\mathbf{y}, \mathbf{z})
        \end{equation}
    For a given dynamic range $K$, the shift pseudo-metric defines a pseudo-metric space in which the residue vectors are located. However, it differs from a metric in the following way.
    \begin{cor} Let $X_0, Y_0\in\mathcal{M}(K)$ be the largest integers which satisfy $X_0\le X$ and $Y_0\le Y$, respectively. Then $\rho(\mathbf{x}, \mathbf{y})=0$ if and only if $X_0=Y_0$.
            \begin{proof}
            If $X_0=Y_0=km_t$, let $Y=X+r$ with $r\ge 0$ without loss of generality. It is not hard to observe that $\lc Y\rc_{m_l}=\lc X+r \rc_{m_l}=\lc X\rc_{m_l}+\lc r \rc_{m_l}$, for any $l$. Combining it with the definition of $\rho$, we have $\rho(\mathbf{x}, \mathbf{y})=0$. On the other hand, if $\rho(\mathbf{x}, \mathbf{y})=0$, we have
            \begin{equation}
                \forall 1\le l<j\le L,  \lc X \rc_{m_l} - \lc X\rc_{m_j} = \lc Y\rc_{m_l} - \lc Y\rc_{m_j}
            \end{equation}
            Note that
            \begin{equation} \label{cor1 eq1}
                \lc X_0\rc_{m_l}-\lc X_0\rc_{m_j}=\lc Y \rc_{m_l} - \lc Y \rc_{m_j}
            \end{equation}
            and
            \begin{equation} \label{cor1 eq2}
                \lc Y_0\rc_{m_l}-\lc Y_0\rc_{m_j}=\lc X\rc_{m_l} - \lc X\rc_{m_j}
            \end{equation}
            Subtract (\ref{cor1 eq1}) by (\ref{cor1 eq2}) then we obtain $\forall l$
            \begin{equation}\label{assign}
                \lc X_0\rc_{m_l} - \lc Y_0\rc_{m_l} = C
            \end{equation}
            where $C$ is a constant. Since $X_0,Y_0\in\mathcal{M}(K)$, assume $m_g|X_0$ and $m_h|Y_0$. If $g\ne h$,  with assigning $l=g$, $l=h$ in equation (\ref{assign}) respectively, we have $C=-\lc Y_0\rc_{m_h}\le 0$ and $C=\lc X_0\rc_{m_g} \ge 0$ , which leads to $C=0$.  Otherwise if $g=h$, by assigning $l=g$, we have $C=0$ as well. Thus $X_0=Y_0$.
            \end{proof}
    \end{cor}
    The corollary above shows that the shift pseudo-metric divide $[0, K]$ into $|\mathcal{M}(K)|$ classes, each represented by an element in $\mathcal{M}(K)$. Furthermore, the shift pseudo-metric acts as a metric over $\mathcal{M}(K)$ according to the metric axioms.

\subsection{Error Bound}

    \begin{thm}\label{errordist}
        Assuming that each component $r_l$ of the vector $\bf x$ may have an error, i.e.,
        $\tilde r_l = r_l + \Delta r_l$, where $|\Delta r_l|<\delta$, $r_l = \lc X\rc_{m_l}$ and  $\tilde r_l = \lc \tilde X\rc_{m_l}$. Then $\rho(\mathbf{\tilde x}, \mathbf{x}) < 2\delta $, where $\mathbf{\tilde x} $ is the erroneous residue vector.
    \end{thm}
    \begin{proof}
        \begin{equation} \label{thm1 eq}
        \begin{aligned}
            \rho(\mathbf{\tilde x}, \mathbf{x}) &= \max_{1\le l<j\le L} |(\lc X \rc_{m_l} - \lc X\rc_{m_j}) - (\lc \tilde X
            \rc_{m_l} - \lc \tilde X\rc_{m_j})| \\
            &= \max_{1\le l<j\le L} |\Delta r_l - \Delta r_j| < 2\delta
        \end{aligned}
        \end{equation}
    \end{proof}

    \begin{cor}
        Given a dynamic range $K$, the robustness bound is
        \begin{equation}
            \delta = \frac{1}{4} \min_{X, Y\in \mathcal{M}(K), X\not=Y} {\rho(\mathbf{x, y})}
        \end{equation}
    \end{cor}
    \begin{proof}
    For $X \in [0, K]$, there exists the largest $X_0 \in\mathcal{M}(K)$, such that $X_0\le X$. When the robustness bound is $\delta$, we have $\rho(\mathbf{\tilde x, x}) < 2\delta$ by Theorem \ref{errordist}. Therefore we have
        \begin{equation}\label{rho1}
            \rho(\mathbf{\tilde x, x_0}) \le \rho(\mathbf{\tilde x, x}) + \rho(\mathbf{x, x_0}) < 2\delta + 0 = 2\delta
        \end{equation}
        In addition, for $\forall X' \in\mathcal{M}(K), X' \not=X_0$, evidently $\rho(\mathbf{x_0, x'}) \ge 4\delta$ by the definition of $\delta$.
        \begin{equation}
            \begin{aligned}\label{rho2}
                \rho(\mathbf{\tilde x, x'}) &\ge \rho(\mathbf{x_0, x'}) - \rho(\mathbf{x_0, \tilde x})\\ &> 4\delta-2\delta = 2\delta
            \end{aligned}
        \end{equation}
        Thus for each erroneous vector $\mathbf{\tilde x}$, we can correctly recover such $X_0$. This allows us to project the erroneous vector $\mathbf{\tilde x}$ onto the line starting at $\mathbf{x_0}$ with the direction vector $(1,1...1)$. It is easy to check that by conducting CRT on the projected vector $\mathbf{\widehat x}$, we have $|\widehat X - X|<\delta$. Observing the fact that inequalities (\ref{thm1 eq})(\ref{rho1})(\ref{rho2}) cannot be improved, we conclude the robustness bound is exactly $\delta$.
    \end{proof}
    Corollary 2 shows that  the robustness is achieved through figuring out $X_0$. With the formalization of $\delta$ developed, in the following, when we emphasize the relationship between $\delta$ and $K$ in the case that $K$ is given, we denote the minimal distance between residue vectors of any two integers within the dynamic range $[0,K]$ by $4\delta(K)$.
    \begin{lm}\label{lm:toproof 2->1}
         For any $X, Y \in \mathcal{M}(K)$, if there exist $l$ and $j$, such that $|(\lc X \rc_{m_l} - \lc X \rc_{m_j}) - (\lc Y \rc_{m_l} - \lc Y \rc_{m_j})| \ge m_l$,
        then $K \ge m_l$.

    \end{lm}
    \begin{proof}
        Without loss of generality, assume $X>Y$, then $\lfloor X/m_l\rfloor \ge \lfloor Y/m_l\rfloor$ and $\lfloor X/m_j\rfloor \ge \lfloor Y/m_j\rfloor$. Thus
        \begin{equation}
            \begin{aligned}
            m_l &\le |(\lc X\rc_{m_l} - \lc X\rc_{m_j}) - (\lc Y\rc_{m_l} - \lc Y\rc_{m_j})| \\
                &= |(\lfloor X/m_l\rfloor -\lfloor Y/m_l\rfloor)m_l - (\lfloor X/m_j\rfloor -\lfloor Y/m_j\rfloor)m_j| \\
                &\le \max_{k=l,j}\lbrace (\lfloor X/m_k\rfloor -\lfloor Y/m_k\rfloor)m_k \rbrace \\
                &\le \max_{k=l,j}\lbrace \lfloor X/m_k\rfloor m_k \rbrace \le X \le K
            \end{aligned}
        \end{equation}
    \end{proof}

    \begin{thm}\label{2->1}
        \begin{equation}
            \min_{X, Y\in \mathcal{M}(K), Y\not=X}\rho(\mathbf{y, x}) = \min_{X > 0, X\in \mathcal{M}(K)} \rho(\mathbf{x, 0})
        \end{equation}
    \end{thm}
    \begin{proof}
        As it is obvious that  $\min_{X, Y\in \mathcal{M}(K), Y\not=X}\rho(\mathbf{y, x}) \leq \min_{X > 0, X\in \mathcal{M}(K)} \rho(\mathbf{x, 0})$, hence we only need to prove that $\forall X, Y\in \mathcal{M}(K), X > Y$,
        \begin{equation}
            \rho(\mathbf{y, x}) \ge \min_{X > 0, X\in \mathcal{M}(K)} \rho(\mathbf{x, 0})
        \end{equation}
        Noticing that $\rho(\mathbf{y, x})$ can be expressed as
        \begin{equation}
            \max_{1\le l<j\le L} |(\lfloor \frac {X}{m_l}\rfloor -\lfloor \frac {Y}{m_l}\rfloor)m_l - (\lfloor \frac {X}{m_j}\rfloor -\lfloor \frac {Y}{m_j}\rfloor)m_j)|
        \end{equation}
        and if there exists an integer $Z$, $Z \in\mathcal{M}(K)$, $Z>0$, which satisfies $ \lfloor Z/m_l\rfloor = \lfloor X/m_l\rfloor - \lfloor Y/m_l\rfloor$ for $\forall l$,  we have
        \begin{equation}
        \begin{aligned}
             \rho(\mathbf{y, x}) &= \max_{1\le i<j\le L} |\lfloor Z/m_l\rfloor m_l - \lfloor Z/m_j\rfloor m_j| \\
             &= \rho(\mathbf{z, 0}) \ge \min_{X > 0, X\in \mathcal{M}(K)} \rho(\mathbf{x, 0})
        \end{aligned}
        \end{equation}
        The argument above shows that if such $Z$ exists, the inequality (17) holds and the theorem holds as well. Now we prove that if for any $l, j$
        \begin{equation}
             (\lfloor X/m_l\rfloor - \lfloor Y/m_l\rfloor)m_l < (\lfloor X/m_j\rfloor - \lfloor Y/m_j\rfloor)m_j + m_j
        \end{equation}
        then $Z$ exists.
        In such case, we can select $Z$ as
        \begin{equation}
            Z = \max_l \lbrace(\lfloor X/m_l\rfloor - \lfloor Y/m_l\rfloor)m_l\rbrace
        \end{equation}
        With the fact that $X>Y$ and $X,Y\in\mathcal{M}(K)$, $Z>0$, there exists $l_0$ such that
        \begin{equation}
            Z = (\lfloor X/m_{l_0}\rfloor - \lfloor Y/m_{l_0}\rfloor)m_{l_0} \le \lfloor X/m_{l_0}\rfloor m_{l_0} \le X \le K
        \end{equation}
        Besides, by the definition of $Z$, for each $l$, we have
        \begin{equation}
            (\lfloor X/m_l\rfloor - \lfloor Y/m_l\rfloor)m_l \le Z < (\lfloor X/m_l\rfloor - \lfloor Y/m_l\rfloor)m_l + m_l
        \end{equation}
        which means $\lfloor Z/m_l\rfloor = \lfloor X/m_l \rfloor - \lfloor Y/m_l \rfloor$ for all $l$. Therefore the selected $Z$ is satisfied. On the other hand, the condition above, which guarantees such an $Z$ exists, can be simplified to
                \begin{equation}
            (\lc X\rc_{m_j} - \lc X\rc_{m_l}) - (\lc Y\rc_{m_j} - \lc Y\rc_{m_l}) < m_j ~~ \forall l, j
        \end{equation}
        Otherwise, if there exist such $l_0$ and $j_0$ such that
        \begin{equation}
            (\lc X\rc_{m_{j_0}} - \lc X\rc_{m_{l_0}}) - (\lc Y\rc_{m_{j_0}} - \lc Y \rc_{m_{l_0}}) \ge m_{j_0}
        \end{equation}
        then we obtain $K\ge m_{j_0}$ by Lemma \ref{lm:toproof 2->1}. Denoting $\min_l\{m_l\}$ as $m_0$, then
        \begin{equation}
            \min_{X>0, X\in \mathcal{M}(K)}\rho(\mathbf{x, 0}) \le \rho(\mathbf{m_0, 0}) = m_0
        \end{equation}
        Hence
        \begin{equation}
        \begin{aligned}
            \rho(\mathbf{y, x})
            &\ge |(\lc X\rc_{m_{l_0}} - \lc X\rc_{m_{j_0}}) - (\lc Y\rc_{m_{l_0}} - \lc Y\rc_{m_{j_0}})|\\
            &\ge m_{j_0} \ge m_0 \ge \min_{X>0, X\in \mathcal{M}(K)}\rho(\mathbf{x, 0})
        \end{aligned}
        \end{equation}
        \end{proof}
          To make a brief summary, we obtain an error bound $\delta(K)$ for the dynamic range $K$ in Corollary 2 and develop a series of reduction to further simplify it. In the above theorem, we prove that the minimal distance between any two residue vectors is equivalent to that between the residue vector of $\mathcal{M}(K)$ and $\mathbf{0}$. In the following, we take a closer look at $\min_{X > 0, X\in \mathcal{M}(K)} \rho(\mathbf{x, 0})$ and show the relationship between $\delta(K)$ and the components of residue vectors.
        \begin{lm}\label{better formula lm}
            For any $km_t \in \mathcal{M}(K)$
            \begin{equation}
                \rho(\mathbf{km_t, 0}) = \max_l \lc km_t\rc_{m_l}
            \end{equation}
            where $\mathbf{km_t} $ denotes the residue vector of $km_t$.
        \end{lm}
        \begin{proof}
            Note that
            \begin{equation}
                \rho(\mathbf{km_t, 0}) = \max_{1\le l<j\le L}|\lc km_t\rc_{m_l} - \lc km_t\rc_{m_j}|
            \end{equation}
            and
            \begin{equation}
                \max_{1\le l<j\le L; l,j\ne t}|\lc km_t\rc_{m_l}-\lc km_t\rc_{m_j}| \le \max_{l\ne t}\lc km_t\rc_{m_l}
            \end{equation}
            Under the case of $l=t$ or $j=t$ or the case of $l,j\ne t$, the above lemma holds.
        \end{proof}
      Lemma \ref{better formula lm} also indicates that when a common factor $\Gamma$ is introduced in each modulus $m_l$, i.e., ${m_l}'=\Gamma m_l$, then $\lc k\Gamma m_t\rc_{\Gamma m_l} = \Gamma \lc km_t\rc_{m_l}$. Therefore, the minimal distance $\min_{X, Y\in \mathcal{M}(K), Y\not=X}\rho(\mathbf{y, x})$ with the moduli ${m_l}'$ is $\Gamma$ times larger than that with the moduli ${m_l}$.
                \begin{thm} \label{1-d form theorem}
            Denoting the minimum modulus as $m_0$, then
            \begin{equation}
                \delta(K) = \min_{m_0\le X\le K}\max_{l} \frac {\lc X\rc_{m_l}}{4}
            \end{equation}
        \end{thm}
        \begin{proof}
           Based on the definition of $\delta(K)$, recalling Theorem \ref{2->1} and Lemma \ref{better formula lm} proposed above, we have
            \begin{equation}
                \delta(K) = \min_{X>0, X\in \mathcal{M}(K)}\frac {\rho(\mathbf{x, 0})}4 = \min_{X>0, X \in \mathcal{M}(K)}\max_{l} \frac {\lc X\rc_{m_l}}{4}
            \end{equation}
            For each integer $X, m_0\le X\le K$, there exists the largest integer $X_0 \in \mathcal{M}(K)$, $0<X_0=km_{l_0}\le X$. Then we have
            $\lc X\rc_{m_{l_0}} \le \lc X\rc_{m_l}$ for every $l \ne l_0$. Therefore,
            \begin{equation}
            \begin{aligned}
                \max_{l}\lc X\rc_{m_l} &= \max_{l\ne l_0}\lc X\rc_{m_l} \ge \max_{l\ne l_0}\lc km_{l_0}\rc_{m_l} \\ &= \max_{l} \lc X_0 \rc_{m_l}
                \ge \min_{Y>0, Y \in \mathcal{M}(K)}\max_{l}\lc Y\rc_{m_l}
            \end{aligned}
            \end{equation}
            Hence
            \begin{equation}
                \min_{m_0\le X \le K}\max_{l} \frac {\lc X \rc_{m_l}}{4} \ge \min_{Y>0, Y \in \mathcal{M}(K)}\max_{l} \frac {\lc Y \rc_{m_l}}{4}
            \end{equation}
            On the other hand, it's obvious that
            \begin{equation}
                \min_{m_0\le X \le K}\max_{l} \frac {\lc X \rc_{m_l}}{4} \le \min_{Y>0, Y \in \mathcal{M}(K)}\max_{l} \frac {\lc Y \rc_{m_l}}{4}
            \end{equation}
        \end{proof}
        \begin{cor}\label{probab cor1}
            With a given error bound $\delta$, the dynamic range $K$ is the largest integer $x$ such that for each integer $X\in[m_0, x]$, $\exists l, \lc X\rc_{m_l}\ge 4\delta$.
        \end{cor}

        \begin{rmk}\label{rmk-expression}
            The equation (31) has an alternate form as follows,
            \begin{equation}
                4\delta(K) = \min_{\min(\mathscr{M})=m_0\le x\le K}||\mathbf{x}||_\infty
            \end{equation}
            From the equation (36), it is easy to verify that $4\delta(K)$ is a positive integer when $K<lcm(\mathscr{M})$. In addition, $\delta(K)$ is a discrete monotone decreasing function of $K$. We denote $\delta_n$ as the $n^{th}$ value as $K$ increases. Correspondingly, $K_n$ is the minimal integer such that $\delta(K_n-1)=\delta_{n}$.
        \end{rmk}

        \textbf{Example:} When $m_1 = 11\times 15, m_2 = 11\times 31, m_3 = 11\times 24$, the maximum dynamic range $K$ is $40920$. Then
        \begin{align*}
        K_1 &= 165, \delta_1 = 165/4; K_2 = 341, \delta_2=77/4\\
            K_3 &= 1056, \delta_3=66/4;K_4=1364, \delta_4=44/4\\K_5 &= 4785, \delta_5=33/4;K_6=10571,\delta_6=11/4
        \end{align*}
        \begin{cor}\label{rmk x dim}
            From Corollary \ref{probab cor1}, $K$ is the smallest integer such that $\forall l, \lc K+1\rc_{m_l} < 4\delta$. Hence for each $l$, there exists a multiple of $m_l$ in the set $\{(K+1)-4\delta+1, (K+1)-4\delta+2,...,(K+1)-1,K+1\}$. The inverse is also true, i.e., $K$ is the smallest integer such that the set above contains at least one multiple of each modulus.
        \end{cor}
        \begin{rmk}
             Following the idea of Corollary \ref{rmk x dim}, we can estimate the dynamic range $K$. A lower bound of $K$ is $\sqrt[4\delta]{lcm(\mathscr{M})}$. Moreover, when all of the moduli are primes, $p_1, p_2, ..., p_L$, $K$ is upper bounded by $\frac {p_1p_2p_3...p_L}{(1+2\delta)^L}$.
        \end{rmk}

\end{section}
\begin{section}{Further Analysis under the Case of Two Moduli}
\noindent When $card(\mathscr{M}) = 2$, a simple recursive formula for $\{\delta_n\}$ and $\{K_n\}$ can be derived. With the analysis for general $L$ above, we give a short proof to the conclusions obtained in \cite{2017towards} and point out how to generalize the encoding system. First, we give a lemma below.
    \begin{lm} \label{2-d lemma}
        When $card(\mathscr{M})=2$ and $n\ge 2$, if $\rho(\mathbf{k_1m_1,k_2m_1})<\delta_{n-1}$, then, $|k_1-k_2|m_1\ge K_n-\delta_n$.
    \end{lm}
    \begin{proof}
        Rewrite $k_1m_1$ and $k_2m_1$ as, $k_1m_1 = q_1m_2 + \lc k_1m_1\rc_{m_2}$ and $k_2m_1 = q_2m_2 +\lc k_2m_1\rc_{m_2}$. Without loss of generality, assume that $k_1>k_2$, hence $q_1\ge q_2$. Then we have
        \begin{equation}
            (k_1-k_2)m_1 = (q_1-q_2)m_2+\lc k_1m_1\rc_{m_2}-\lc k_2m_1\rc_{m_2}
        \end{equation}
        \textbf{Case 1} If $\lc k_1m_1\rc_{m_2}\ge \lc k_2m_1\rc_{m_2}$, then
        \begin{equation}
        \begin{aligned}
            \delta_{n-1} > |\lc k_1m_1\rc_{m_2}-\lc k_2m_1\rc_{m_2}|&=\lc(k_1-k_2)m_1\rc_{m_2}\\&=\rho(\mathbf{(k_1-k_2)m_1, 0})
        \end{aligned}
        \end{equation}
        According to Theorem \ref{2->1} and Remark \ref{rmk-expression}, if $(k_1-k_2)m_1<K_n$, then $\rho(\mathbf{(k_1-k_2)m_1, 0})\ge \delta_{n-1}$, which leads to contradiction. Hence we have $(k_1-k_2)m_1\ge K_n\ge K_n-\delta_n$. \\
       \textbf{Case 2} If $\lc k_1m_1\rc_{m_2}< \lc k_2m_1\rc_{m_2}$, we have
               \begin{equation}
                   |\lc k_1m_1\rc_{m_2}- \lc k_2m_1\rc_{m_2}|=m_2-\lc(k_1-k_2)m_1\rc_{m_2}<\delta_{n-1}
               \end{equation}
               Denote $(\lfloor\frac {(k_1-k_2)m_1}{m_2}\rfloor +1)m_2$ as $A$. Note that
               \begin{equation}
               \begin{aligned}
                   \rho(\mathbf{a, 0}) &=\lc(k_1-k_2)m_1-\lc(k_1-k_2)m_1\rc_{m_2}+m_2\rc_{m_1}\\
                   &=\lc m_2-\lc(k_1-k_2)m_1\rc_{m_2}\rc_{m_1}\\
                   &= m_2-\lc(k_1-k_2)m_1\rc_{m_2}<\delta_{n-1}
               \end{aligned}
               \end{equation}
               With a similar argument in $\textbf{Case 1}$, we have $A \ge K_n$ and
               \begin{equation}
                   (k_1-k_2)m_1=A-\rho(\mathbf{a,0}).
               \end{equation}
               Considering the function $f(X)=X-\rho(\mathbf{x,0})$, it is not hard to verify that $f(km_1)$ and $f(km_2)$ are both monotone increasing functions of $k$.
               Therefore
               \begin{equation}
                   (k_1-k_2)m_1\ge K_n-\rho(\mathbf{K_n, 0})=K_n-\delta_n
               \end{equation}
    \end{proof}
    Without loss of generality assuming that $m_2 > m_1$, it is not hard to figure out that $\delta_0 = m_2$, $\delta_1=m_1$, $K_1=m_1$ and $K_2=m_2$. Let $K_0=m_2$ and we have the following theorem.
    \begin{thm}
       For $n =1,2,3, ...$
      \begin{equation}
       \delta_{n+1} = \lc \delta_{n-1}\rc_{\delta_n}
    \end{equation}
    \begin{equation}
       K_{n+1} = K_{n-1} + (K_n-\delta_n)\lfloor \frac{\delta_{n-1}}{\delta_n}\rfloor
    \end{equation}
    \end{thm}
    \begin{proof}
        Denote $\rho(\mathbf{x, 0})$ as $\rho(X)$ for simplicity and we prove the theorem by induction on $X$, together with arguments below:
        \begin{enumerate}
            \item $m_1|(K_{2j}-\delta_{2j})$ and $m_2|(K_{2j+1}-\delta_{2j+1})$, $j \in \mathbb{Z^*}$
            \item $m_2|K_{2j}$ and $m_1|K_{2j+1}$, $j \in \mathbb{Z^*}$
        \end{enumerate}
        Noticing that $\rho(K_{2j})=|\lc K_{2j}\rc_{m_2}-\lc K_{2j}\rc_{m_1}|$ and $\rho(K_{2j+1})=|\lc K_{2j+1}\rc_{m_2}-\lc K_{2j+1}\rc_{m_1}|$, by argument 2), we have
        \begin{equation}
            \lc K_{2j}\rc_{m_1}=\delta_{2j}\phantom{y}and\phantom{y}\lc K_{2j+1}\rc_{m_2}=\delta_{2j+1}
        \end{equation}
        Assume that, when $n\le j-1$, the arguments above are true.

        \textbf{Case 1} $j=2t$. In this case, we have $m_2|K_j=K_{2t}$ and $m_1|K_{j-1}=K_{2t-1}$. By Theorem \ref{2->1}, we only have to check those $X$ such that $\rho(X) < \delta_{2t}$. Define $\tilde\rho(X) := \lc X\rc_{m_2}-\lc X\rc_{m_1}$. Noting that when $m_1|X$, $\tilde\rho(X)>0$ and $m_2|X$, $\tilde\rho(X)<0$, $\tilde\rho(K_{2t-1})=\delta_{2t-1}$. First, we determine $s$ that satisfies $\tilde\rho(K_{2t-1}+sm_1)\in(-\delta_{2t}, \delta_{2t-1})$.
        \begin{equation}
        \begin{aligned}
            \delta_{2t-1}>\tilde\rho(K_{2t-1}+sm_1)&=\lc K_{2t-1}+sm_1\rc_{m_2}\\
            &=\lc \delta_{2t-1} + \lc sm_1\rc_{m_2}\rc_{m_2}
        \end{aligned}
        \end{equation}
        Note that $\delta_{2t-1} \le \delta_{2t-1}+\lc sm_1\rc_{m_2}<2m_2$, hence
        \begin{equation}
            \tilde\rho(K_{2t-1}+sm_1) = \delta_{2t-1}+\lc sm_1\rc_{m_2}-m_2>0
        \end{equation}
        Assuming that $s$ is increasing, we consider two "adjacent" $s_g$ and $s_{g+1}$ such that,
        \begin{equation}
            m_2-\delta_{2t-1}<\lc s_tm_1\rc_{m_2}<m_2\phantom{y} t=g,g+1
        \end{equation}
        Consequently we have
        \begin{equation}
            \rho(\mathbf{s_gm_1, s_{g+1}m_1})=|\lc s_gm_1\rc_{m_2}-\lc s_{g+1}m_1\rc_{m_2}|<\delta_{j-1}
        \end{equation}
        By Lemma \ref{2-d lemma},
        \begin{equation}
            (s_{g+1}-s_{g})m_1\ge K_j-\delta_j
        \end{equation}
        On the other hand, when $t\le \lfloor \delta_{2t-1}/\delta_{2t}\rfloor$,
        \begin{equation}
            \tilde\rho(K_{2t-1}+t(K_{2t}-\delta_{2t}))=\lc\delta_{2t-1}-t\delta_l\rc_{m_2}=\delta_{2t-1}-t\delta_{2t}<\delta_{2t-1}
        \end{equation}
        This leads to $sm_1=t(K_{2t}-\delta_{2t})$ when $t\le \lfloor \delta_{2t-1}/\delta_{2t}\rfloor$.\\
        \indent Second, we consider $\tilde\rho(K_{2t}+sm_2)\in(-\delta_{2t},\delta_{2t-1})$.  If $t\le \lfloor \delta_{2t}/\delta_{2t+1}\rfloor$, then $sm_2=t(K_{2t+1}-\delta_{2t+1})$.
        Therefore, if $s>0$, then
        \begin{equation}
            K_{2t}+sm_2\ge K_{2t}+K_{2t+1}-\delta_{2t+1}>K_{2t+1}
        \end{equation}
        However, when $K_{2t} \le X \le K_{2t+1}$, equation (48) does not hold. Thus $\delta_{2t+1}$ is obtained exactly when $t=\lfloor \delta_{2t-1}/\delta_{2t} \rfloor$ and
        \begin{equation}
            K_{2t+1}=K_{2t-1}+(K_{2t}-\delta_{2t})\lfloor \delta_{2t-1}/\delta_{2t} \rfloor
        \end{equation}
        along with
                        \begin{equation}
            \delta_{2t+1}=\lc \delta_{2t-1}\rc_{2t}
        \end{equation}
        Based on the induction hypothesis, $m_2|K_{2t+1}$. Similarly,
        \begin{equation}
            K_{2t+1}-\delta_{2t+1} = (K_{2t-1}-\delta_{2t-1}) + K_l\lfloor \delta_{2t-1}/\delta_{2t} \rfloor
        \end{equation}
        Still by the induction hypothesis, $m_1|K_{2t+1}-\delta_{2t+1}$ is clear.\\
        \textbf{Case 2} When $j=2t-1$, consider those $\tilde\rho(X)\in(-\delta_{j-1}=-\delta_{2t-2}, \delta_j=\delta_{2t-1})$, a similar discussion can lead to the conclusion.
        Then the induction is completed.
    \end{proof}

    \begin{rmk}
        Corollary \ref{rmk x dim} shows that finding $K$ under given $\delta$  is equivalent to finding smallest non-negative integers $k_1, k_2, ..., k_L,$ such that $\forall l, j, |k_lm_l-k_jm_j|<4\delta$. Thus determining the closed-form expression of $K$s and $\delta$s in $L$-dimensional space is to solve a combination of $\frac{L(L-1)}{2}$ 2-D cases entangled together.
            \end{rmk}
\end{section}

\begin{section}{Reconstruction Scheme \& Modulus Selection}
 \noindent In this section, we develop a robust reconstruction scheme and provide a probability analysis of the modulus selection that can achieve the given error bound. Given the encoding system with the prescribed $\delta$ and dynamic range $[0,K]$, we consider recovering $E_\alpha$ by CRT from each possible error vector $\mathbf{e_\alpha}=(\alpha_1 ,\alpha_2...,\alpha_{\tau} )$, where $\alpha_j \in \{\lceil-\delta+1 \rceil,\lceil-\delta+2 \rceil,...,\lceil  \delta-2\rceil, \lceil \delta-1 \rceil \}$. Store $\{E_\alpha\}$ in a list $\mathcal{L}_{e}$ and the length of $\mathcal{L}_{e}$ is $\tau = (2\delta-1)^{L}$. \footnote {Another scheme is to store those starting points of the lines, which the residue representation of $E_\alpha+\delta$ lie on. An error list of length at most $L(2\delta-1)^{L-1}$ is necessary.} Based on $\mathcal{L}_{e}$, we propose a search-based reconstruction scheme as follows. \footnote{Lattice-based decoding schemes for remainder codes have been studied in \cite{2000chinese}, \cite{2004lattice}. Similar techniques can also be generalized in our case. However, those algorithms are more or less restricted due to the approximation factor of LLL Reduction. The reconstruction scheme we proposed is general, based on the algebra property of the remainder code.}
    \begin{algorithm}[h]
        \caption{}
        \textbf{Input} Given the erroneous residue vector $\mathbf{\tilde x}$\\
        Step-1. Compute $\tilde X$ by CRT from erroneous vector $\mathbf{\tilde x}$.\\
        Step-2. Implement a binary search for all $E_{\alpha_j}$ in $\mathcal{L}_{e}$, such that $0\le\tilde X-E_{\alpha_j} \le K$, $j=1,2,...,\vartheta$.\\
        Step-3. Calculate $\widehat X = \sum_{j=1}^{\vartheta}(\tilde X-E_{\alpha_j})/\vartheta$.\\
        \textbf{Output} Output $\widehat X$.
    \end{algorithm}
        \\ \indent The time complexity is $O(Llog(2\delta))$ with $O((2\delta)^L)$ space complexity. In the following, we give a proof to show the validity of the proposed algorithm and $|\widehat X-X|<\delta$, where $X$ is the original integer without errors.\\
    \indent    Assuming that $X_0$ is the largest integer in $\mathcal{M}(K)$ such that $X_0 \le X$, consider a sphere under the shift pseudo metric, of which the center is in the position of vector $\mathbf{\tilde x}$ with radius $2\delta$. Recalling equation (\ref{rho1}) and (\ref{rho2}), we know that there is only one line in the $L$-dimensional space that intersects with the sphere, which is exactly the line starting at $\mathbf{x_0}$ whose direction vector is $(1,1,...1)$. Denote the line as $l_{X_0}$. Based on Theorem \ref{errordist}, we have $\rho(\mathbf{\tilde x - e_{\alpha_j}}, \mathbf{\tilde x}) < 2\delta.$
        Hence if $0\le\tilde X - E_{\alpha_j} < K$, then $\mathbf{\tilde x - e_{\alpha_j}}$ is on the line $l_{X_0}$. Furthermore, the original $\mathbf{x}$ is on $l_{X_0}$ as well and $\rho(\mathbf{x,\tilde x})<2\delta$. Thus at least one $E_{\alpha_j}$ can be found in step 2. Then by step 3, $\widehat X$ is exactly the projection of $\tilde X$ onto $l_{X_0}$. Further calculation leads to $|\widehat X-X|<\delta$.

        To end this section, we present a lower bound of probability for a random modulus selection, which can achieve the given error bound $\delta$ and dynamic range $K$ by randomly selecting $L$ primes within interval $[2^{\beta-1}, 2^{\beta}]$ as moduli. We will prove that with suitable $L$ and ${\beta}$, the successful probability is extremely large. Besides, we will demonstrate how to select such kind of $L$ and ${\beta}$. The following results are derived from Corollary 3.
          \begin{thm}
         When all moduli are primes, $p_1, p_2,...,p_L$, the dynamic range $K$ is $x-1$ where $x$ is the smallest positive integer such that $p_1p_2...p_l | (x-4\delta+1)(x-4\delta+2)\cdots(x-1)x,x\ge 4\delta$, with the error bound $\delta$ given.
     \end{thm}
     Without loss of generality, assume $p_1<p_2<...<p_L$.
     \begin{thm}
         With given $K$ and $\delta$, the probability of success of random modulus selection is no less than
         \begin{equation}
             1- \lfloor \frac{K}{p_L}\rfloor(\frac{4\delta log_2K}{2^{{\beta}-1}})^L
         \end{equation}
     \end{thm}
     \begin{rmk}
         If $4\delta log_2K<2^{{\beta}-1}$, the probability of failure is with exponential decay in $L$. Besides, it also decreases sharply as ${\beta}$ increases. Hence $L$ and ${\beta}$, which are not large, can be chosen to ensure that the success probability above is acceptable. Once chosen, verifying these moduli is trivial in complexity $O((4\delta)^{L-1})$.
     \end{rmk}
        A tighter bound can be further derived as follows.
     \begin{thm}
         The probability of success is no less than
         \begin{equation}
             1-\frac{(4\delta)^{L+1}(\mathbf{\Gamma}(L+1, a)-\mathbf{\Gamma}(L+1, b)}{p_Llog^L(2)2^{L({\beta}-1)}(-1)^L}
         \end{equation}
         where $a = -log((\lfloor K/p_L\rfloor+1) p_L)$, $b = -log(p_L)$ and $\mathbf{\Gamma}(x,y)$ is the Gamma function.
     \end{thm}
     \begin{proof}
         We prove Theorem 6 \& 7 together. From Theorem 5, it is a necessary condition that $p_1p_2...p_L \nmid (X-4\delta+1)(X-4\delta+2)\cdots(X-1)X$ when $X\le K$. Denote the interval $[a-4\delta+1, a]$ as $I(a,\delta)$. Consider those intervals which contain a multiple of $p_L$ and there exists $4\delta$ such intervals which include $kp_L$, where $k$ is an integer smaller than $\frac{K}{p_L}$ and the integers in such intervals are approximately $log_2(kp_L)$-bit long. Therefore, each integer can have at most $log_2(kp_L)/(\beta-1)$ prime factors in $[2^{{\beta}-1},2^{\beta}]$. Then the whole interval contains approximately $4\delta log_2(kp_L)/({\beta}-1)$ prime factors at most. Thus the number of selection methods is approximated as
     \begin{equation}\label{up method}
         \sum_{k=1}^{\lfloor K/p_L\rfloor} 4\delta \binom{\frac{4\delta log_2(kp_L)}{{\beta}-1}}{L} < \frac{4\delta}{L!} \sum_{k=1}^{\lfloor K/p_L\rfloor} (\frac{4\delta log_2(kp_L)}{{\beta}-1})^L
     \end{equation}
     On the other hand, there are at least $2^{{\beta}-1}/{\beta}$ primes in the interval $[2^{{\beta}-1},2^{\beta}]$. Therefore the number of selection is approximately as \footnote{ Note that $2^{\beta}$ is far larger than $L$ in real implementation.}

     \begin{equation}
          \label{down method}
         \binom{2^{{\beta}-1}/{\beta}}{L} \approx \frac{1}{L!} (\frac{2^{{\beta}-1}}{{\beta}})^L
     \end{equation}
     The theorem is obtained by using indefinite integral of $log^L_2(x)$ to estimate (\ref{up method}), and dividing it by (\ref{down method}).
 \end{proof}
     \begin{cor}
         By our previous analysis, if these moduli are multiplied by a common factor $\Gamma$, then $\forall j$, both $K_j$ and $\delta_j$ are multiplied by $\Gamma$ as well. In such case, the problem is simplified to find a set of moduli to achieve an expected dynamic range $[0,\frac{K}{\Gamma}]$ and error bound $\frac{\delta}{\Gamma}$.

     \end{cor}
\end{section}

\section{Robust CRT for Multiple Integers}

\noindent In this section, we will present our main contribution, the robust CRT for multiple integers (Generalized Robust CRT, GRCRT) . As mentioned before, developing such GRCRT for multiple integer is to overcome the hardness from lacking correspondence and the errors caused by noise introduced. Therefore, we will propose the complete solution for the problem in two steps to achieve robustness and figure out correspondence, respectively.

\subsection{Reduction to GCRT}
Assume that $m_l=\Gamma M_l$ for $l=1,2,...,L$, where $\{M_l\}$ are co-prime to each other. Without loss of generality, let $\{m_l\}$ be arranged in an ascending order, so dose $\{M_l\}$. We define $r^{c}_i=\langle X_i \rangle _{\Gamma}$ as the common residue of $X_i$ and $\{r^{c}_i,i=1,2,...,N\}$ are arranged in an ascending order. Let $r_{il}=\langle X_i \rangle _{m_l}$ and the erroneous residue $\widetilde{r}_{il}=\langle r_{il}+\Delta r_{il} \rangle_{m_l}$, since the noises bring errors into the residues, where $\Delta r_{il}$ is the error of $r_{il}$. When $\widetilde{r}_{il}=\widetilde{r}_{ik}$ for $l \neq k$, they are called repeated residues. Let $\widetilde{r}^{c}_{il}=\langle \widetilde{r}_{il} \rangle_{\Gamma} $ and assume  $\Lambda=\{ \widetilde{r}^{c}_{il}, i=1,2,...,N, l=1,2,...,L \}$ to be arranged in an ascending order denoted by $\{\gamma_1,\gamma_2,...,\gamma_{\kappa}\}$, where $\kappa \leq NL$ \footnote{In \cite{2017notes}, it is proved that the probability of existing repeated residues  is greatly small. However, for the completeness of analysis, we still take such case into consideration and thus assume $\kappa \leq NL$. }, with equality if no repeated residues exist in $\Lambda$. We define
 $$d_\Gamma(X,Y)=\min_{k \in \mathbb{Z}} \{|X-Y+k\Gamma|\}$$
 for any integers $X$ and $Y$. As we already assume that $r^{c}_i$ are in an ascending order, let
 $$D_i=D(r^{c}_{i}, r^{c}_{i+1})=r^{c}_{i+1}-r^{c}_{i}$$
  when $i \in  \{1,2,...,N-1\}$, and
$$D_N=D(r^{c}_{N}, r^{c}_{1})=r^{c}_{1}-r^{c}_{N}+\Gamma$$
when $i=N$.

When the error bound of $\{\Delta r_{il}\}$ satisfies that $\delta< \frac{\Gamma}{4N}$, we claim that there is a solution to determine $\{X_{i}\}$ such that $|X_i-\widetilde{X}_{i}| < \delta$, where $\{\widetilde{X}_{i}\}$ are estimation of $X_i$ correspondingly. First we propose the following lemma.

\begin{lm} \label{circle} There exists $\xi \in \{1,2,...,\kappa \} $ such that
\begin{equation}
\label{feiyang}
\gamma_{\langle \xi+1 \rangle_{\kappa}}-\gamma_{\xi}+\Gamma \mathbf{1}( \xi=\kappa ) > 2\delta
\end{equation}
where $\mathbf{1}( \xi=\kappa )$ is an indicator, which turns to be 1 iff $\xi=\kappa$. Otherwise $\mathbf{1}( \xi=\kappa )=0$, when $\xi \not =\kappa$.

\end{lm}

\begin{proof} Firstly, assume that there exists some $i_0$ such that $D( r^{c}_{i_0}, r^{c}_{\langle i_0+1 \rangle_{N}})>4\delta$. We can find a $\xi \in \{1,2,...,\kappa \}$ such that $d_\Gamma (\gamma_{\xi},r^{c}_{i_0}) < \delta$, $d_\Gamma(\gamma_{\langle {\xi}+1 \rangle_{\kappa}}, r^{c}_{\langle i_0+1 \rangle_{N}} ) < \delta$ and $d_\Gamma (\gamma_{\xi},\gamma_{\langle {\xi}+1 \rangle_{\kappa}}) > 2\delta $, since $D( r^{c}_{i_0}, r^{c}_{\langle i_0+1 \rangle_{N}})>4\delta$. Particularly, when $\xi=\kappa$, it gives $\gamma_{1}-\gamma_{\kappa}+\Gamma> 2\delta$. Otherwise, if no such $\xi$ exists satisfying equation (\ref{feiyang}), then $D ( r^{c}_i, r^{c}_{\langle i+1 \rangle_{N}}) < 4\delta$ for $i\in \{1,2,...,N\}$. However, $\Gamma = \sum_{i=1}^{N} D (r^{c}_i, r^{c}_{\langle i+1 \rangle_{N}}) < 4N\delta < \Gamma$, which leads to a contradiction.
\end{proof}

\begin{rmk} $\frac{\Gamma}{4N}$ is a tight bound of $\delta$. That is to say that it is true by replacing all $\delta$ with $\frac{\Gamma}{4N}$ in (\ref{feiyang}), i.e., there existing some $i \in \{1,2,...,N\}$ such that $D_{i} > 2\frac{\Gamma}{4N}$. Therefore, in the following, we make a slight change of the definition of $\delta$, i,e,, assuming $$\max_{il} |\Delta_{il}| \leq \delta < \frac{\Gamma}{4N}$$ to have a tighter estimation.
\end{rmk}

\indent Although a modulo operation is with periodic properties, the residue of an integer modulo $\Gamma$ is within the range $[0,\Gamma)$. Since the errors are introduced in residues, it is possible to lead to $r^{c}_{i}+\Delta r_{il}< 0$ or $r^{c}_{i}+\Delta r_{il} \geq \Gamma$. In order to make the analysis more conveniently, we first discuss the following two cases.\\
\indent \textbf{Case I}. $\xi=\kappa
$. In this case, $\widetilde{r}^{c}_{il} \in [0,\Gamma-2\delta)$ or $\widetilde{r}^{c}_{il} \in [2\delta, \Gamma)$ for any $i \in \{1,2,...,N\}$. Therefore, define
\begin{equation}
                              \hat{r}^{c}_{il}=\widetilde{r}^{c}_{il}
\end{equation}
\indent \textbf{Case II}. $\xi \not= \kappa
$. In this case, if $\widetilde{r}^{c}_{il} \leq \gamma_{\xi}$ for any $i \in \{1,2,...,N\}$, define
\begin{equation}
\hat{r}^{c}_{il}=\widetilde{r}^{c}_{il}
\end{equation}
otherwise define
\begin{equation}
\hat{r}^{c}_{il}=\widetilde{r}^{c}_{il}-\Gamma
\end{equation}
\indent Inspired by
$$\langle q_i \rangle_{M_l}= \langle \lfloor \frac{X_i}{\Gamma} \rfloor \rangle_{M_l}=  \langle \frac{X_i-r^{c}_i}{\Gamma} \rangle_{M_l}=\langle \frac{r_{il}-r^{c}_i}{\Gamma} \rangle_{M_l},$$ we define
\begin{equation}
\widetilde{q}_{il} = \langle \lfloor \frac{ \widetilde{r}_{il}-\hat{r}^{c}_{il}}{\Gamma} \rfloor \rangle_{M_l}= \langle \frac{ \widetilde{r}_{il}-\hat{r}^{c}_{il}}{\Gamma} \rangle_{M_l}
\end{equation}

\indent With the methods developed in \cite{2016symmetric}, 
here we just assume that $M=\Gamma \prod_{l=1}^{L}M_l$ is big enough, where the explicit lower bound will be given later. Additionally, the cases of $X_i<\delta$ and all $X_i+\Delta r_{il}<0$ are not considered, which will be analyzed in Remark \ref{margin}. Let $q_i=\lfloor \frac{X_i}{\Gamma} \rfloor$ and $\widetilde{q}_i$ is the integer recovered with residues, $\{\widetilde{q}_{il}, l=1,2,...,L\}$, via CRT. We will prove later that $0 \leq |q_i-\widetilde{q_i}| \leq 1$ is held. In this case, the problem of determining multiple integers with RCRT is converted to determine $\widetilde{q}_{i}$ with a generalized error-free CRT (GCRT). One thing worth mentioning is that $\widetilde{r}_{il}$, $i=1,2,...,n$, being distinct for each $l$ is \textbf{not} required. The latest report on repeated residues is given in \cite{2007sharpened}. A GCRT with polynomial running time is proposed for distinct residues in \cite{2016symmetric}. More details about GCRT will be presented in the subsection B. Then we estimate $X_i$ with

\begin{equation}
\widetilde{X}_{i}=[\widetilde{q}_i\Gamma+\frac{\sum_{l=1}^{L} \hat{r}^{c}_{il}}{L}]
\end{equation}
where $[x]= \lfloor x+0.5 \rfloor $. \\
\indent Now we develop the following theorem to prove the robustness of the algorithm.

\begin{thm} \label{GRCRT}
$0 \leq |\widetilde{q}_i-q_i| \leq 1$ and $\alpha_{il}=|(\widetilde{q}_i-q_i)\Gamma+\hat{r}^{c}_{il}-r^{c}_i|=|\Delta r_{il}|$ are satisfied.
\end{thm}

\begin{proof} Before presenting the proof, we first define that, for a specific $i$, $i \in \{1,2,...,N\}$, if there exists an $r^{c}_{i}$ such that some but not all of $\Delta r_{il}, l \in \{1,2,...,L\}$, satisfy $r^{c}_{i}+\Delta r_{il}< 0$ or $r^{c}_{i}+\Delta r_{il} \geq \Gamma$, it is called a $crossing~ residue$.  We divide the proof into two parts, under the condition of \textbf{Case I} and \textbf{Case II} in (61)-(63), respectively.

Under the condition of \textbf{Case I}, there exists no crossing residue, otherwise there exist $\widetilde{r}^{c}_{i\eta_1} \in [0,2\delta]$ and $\widetilde{r}^{c}_{i\eta_2}\in [\Gamma-2\delta, \Gamma]$, where $\eta_1,\eta_2 \in \{1,2,...,L\}$. We have $d_\Gamma (\widetilde{r}^{c}_{i\eta_1},\widetilde{r}^{c}_{i\eta_2})= \widetilde{r}^{c}_{i\eta_1} - \widetilde{r}^{c}_{i\eta_2}+\Gamma \leq 2\delta$ and we also get $\widetilde{r}^{c}_{i\eta_1}-\widetilde{r}^{c}_{i\eta_2}+\Gamma \geq \gamma_{1}-\gamma_{\kappa
}+\Gamma>2\delta$. Due to monotonicity of $\{\gamma_{j}\}$, it results in a contradiction. Therefore, given $i$, there are three subcases for $\Delta r_{il}$, $l \in \{1,2,...,L\}$, \\
Subcase-1: All $\Delta r_{il}$, satisfy $0 \leq r^{c}_i+\Delta r_{il} \leq \Gamma$;\\
Subcase-2: All $\Delta r_{il}$, satisfy $r^{c}_i+\Delta r_{il} \geq \Gamma$;\\
Subcase-3: All $\Delta r_{il}$, satisfy $r^{c}_i+\Delta r_{il}< 0$.\\
\indent Therefore, $\widetilde{q}_{il}$ equals to $\langle q_i \rangle_{M_l}$ for Subcase-1 and $\alpha_{il}=|\widetilde{r}^{c}_{il}-r^{c}_i|=|\Delta r_{il}|$. In Subcase-2, $\widetilde{q}_{il}= \langle q_i+1 \rangle_{M_l}$ and $\widetilde{r}^{c}_{il}={r}^{c}_{i}+\Delta r_{il}-\Gamma$, so $\alpha_{il}=|\Gamma+\hat{r}^{c}_{il}-r^{c}_i |=|\Delta r_{il}|$. Similarly, in Subcase-3, $\widetilde{q}_{il}=\langle q_i-1 \rangle_{M_l}$ and $\widetilde{r}^{c}_{il}={r}^{c}_{i}+\Delta r_{il}+\Gamma$, so $\alpha_{il}=|\Gamma+\hat{r}^{c}_{il}-r^{c}_i |=|\Delta r_{il}|$. Therefore the theorem holds.

Under the condition of \textbf{Case II}, firstly assume that $r^{c}_i$ is not a crossing residue for some $i$. Note that, under the assumption, $\widetilde{r}^{c}_{il}  \in [0,\delta) $ in Subcase-2 and $\widetilde{r}^{c}_{il}  \in (\Gamma-\delta,\Gamma)$ in Subcase-3. Due to $\gamma_{\xi+1}-\gamma_{\xi}>2\delta$, which gives that $\gamma_{\xi}<\Gamma-2\delta$ and $\gamma_{\xi+1}>2\delta$, this leads to that all $\widetilde{r}^{c}_{il} $ in Subcase-2 belong to $[\gamma_{1}, \gamma_{\xi}]$ and all $\widetilde{r}^{c}_{il}$   Subcase-3 belong to $ [\gamma_{\xi+1}, \gamma_{\kappa}]$.

In Subcase-1, it is not hard to verify that all $\widetilde{r}^{c}_{il}  \in [\gamma_{1}, \gamma_{\xi}]$ or all $\widetilde{r}^{c}_{il}  \in [\gamma_{\xi+1}, \gamma_{\kappa }]$, otherwise there exist $\widetilde{r}^{c}_{i\eta_1} \geq \gamma_{\xi+1}$ and $\widetilde{r}^{c}_{i\eta_2} \leq \gamma_{\xi}$, where $\eta_1,\eta_2 \in \{1,2,...,L\}$. However, we have $2\delta \geq \widetilde{r}^{c}_{i\eta_1}-\widetilde{r}^{c}_{i\eta_2} \geq \gamma_{\xi+1}-\gamma_{\xi}>2\delta$, which is a contradiction. According to (62) and (63), $\hat{r}^{c}_{il}$ is an integral shift of $\widetilde{r}^{c}_{il}$, so $0 \leq |q_i-\widetilde{q_i}| \leq 1$ is obvious. Furthermore, since $\widetilde{q}_i={q_i}$ in Subcase-3, we get $\alpha_{il}=|({r}^{c}_{i}+\Delta r_{il}+\Gamma)-\Gamma-r^{c}_i|=|\Delta r_{il}|$. In Subcase-1, when all $\widetilde{r}^{c}_{il} \in [\gamma_{1}, \gamma_{\xi}]$, we get $\widetilde{q}_i={q_i}$ and $\alpha_{il}=|{r}^{c}_{i}+\Delta r_{il}-r^{c}_i|=\Delta r_{il}$. Similarly, $\widetilde{q}_i={q_i}+1$ in subcase-2 and subcase-1 when all $\widetilde{r}^{c}_{il} \in [\gamma_{\xi}, \gamma_{\kappa
}]$ and the same conclusions are achieved. The theorem holds as well in such cases. \\
\indent Next, we discuss the situation where $r^{c}_{i}$ is a crossing residue. If $\widetilde{r}^{c}_{il} \in [\Gamma-2\delta,\Gamma)$, $l \in \{1,2,...,L\}$, it leads to that $\widetilde{r}^{c}_{il} \in [\gamma_{\xi+1}, \gamma_{\kappa
}]$. Similarly, if $\widetilde{r}^{c}_{il} \in [ 0,2\delta )$, then $\widetilde{r}^{c}_{il} \in [\gamma_{1}, \gamma_{\xi}]$. A crossing residue $r^{c}_i$ is within $[ 0,\delta )$ or $ [ \Gamma-\delta,\Gamma ) $. According to the analysis above, it is true that when $r^{c}_{i} \in [\Gamma-\delta,\Gamma)$, there exist $ \widetilde{r}^{c}_{i\theta_1} \in [\gamma_{\xi+1}, \gamma_{\kappa
}] $ and $\widetilde{r}^{c}_{i\theta_2} \in [ \gamma_{1}, \gamma_{\xi}]$, where $\theta_1,\theta_2 \in \{1,2,...,L\}$. According to the relation between an integer and its residue, we have
\begin{equation}
X_i +\Delta_{i\theta_1}= K\Gamma M_{\theta_1}+r_{i\theta_1}+\Delta_{i\theta_1}
\end{equation}
where $K\in \mathbb{Z}$. Then
\begin{equation}
\begin{aligned}
 \frac{\widetilde{r}_{i\theta_1}-\widetilde{r}^{c}_{i\theta_1}}{\Gamma} & = \langle \frac{X_i+\Delta_{i\theta_{1}}-K\Gamma M_{\theta_1}-(\Delta_{i\theta_1}+r^{c}_{i})}{\Gamma} \rangle_{M_{\theta_1}}\\
& =\langle \frac{X_i+r^{c}_{i}}{\Gamma} - KM_{\theta_1}  \rangle_{M_{\theta_1}}=\langle \lfloor \frac{X_i}{\Gamma} \rfloor \rangle_{M_{\theta_1}}
\end{aligned}
\end{equation}
Since $\Gamma-2\delta \leq r^{c}_{i}+\Delta_{i\theta_1}<\Gamma$, we obtain
\begin{equation}
\widetilde{q}_{i\theta_1} = \langle \frac{\widetilde{r}_{i\theta_1}-(\widetilde{r}^{c}_{i\theta_{1}}-\Gamma)} {\Gamma} \rangle_{M_{\theta_{1}}} = \langle \lfloor \frac{X_i}{\Gamma} \rfloor +1\rangle_{M_{\theta_1}}
\end{equation}
Similarly, noticing that $r^{c}_{i}+\Delta_{i\theta_2}>\Gamma$ and $\hat{r}^{c}_{i\theta_2}=\widetilde{r}_{i\theta_2}=r^{c}_{i}+\Delta_{i\theta_2}-\Gamma$, we have
\begin{equation}
\widetilde{q}_{i\theta_2}=\langle \frac{\widetilde{r}_{i\theta_2}-({r}^{c}_{i}+\Delta_{i\theta_2}-\Gamma)}{\Gamma} \rangle_{M_{\theta_2}}=\langle \lfloor \frac{X_i}{\Gamma} \rfloor +1 \rangle_{M_{\theta_2}}
\end{equation}
Thus $\widetilde{q}_{il}=\langle \lfloor \frac{X_i}{\Gamma} \rfloor +1 \rangle_{M_l}$ for $l=1,2,...,L$. For $\widetilde{r}^{c}_{il}=r^{c}_{i}+\Delta r_{il}$, which is within $[\Gamma-2\delta,\Gamma]$, then $ \hat{r}^{c}_{il}=\widetilde{r}^{c}_{il}-\Gamma$ is a negative integer and $\alpha_{il}=|\Gamma+\widetilde{r}^{c}_{il}-\Gamma-r^{c}_{i}|=|\Delta r_{il}|$. For $\hat{r}^{c}_{il}=\widetilde{r}^{c}_{il}=r^{c}_{i}+\Delta r_{il}-\Gamma$ within $[0,\delta]$, $\alpha_{il}=|\Gamma+r^{c}_{i}+\Delta r_{il}-\Gamma-r^{c}_{i|}|=|\Delta r_{il}|$. Correspondingly, when $ r^{c}_{i} \in [0,\delta) $ is a crossing residue, we have $\widetilde{q}_{il} =\langle q_i \rangle_{M_l}$. The conclusion of $\alpha_{il}$ is obtained with similar ideas.
\end{proof}
\indent Based on Theorem \ref{GRCRT}, $|\widetilde{X}_i-X_i| \leq \delta$ is satisfied.

\subsection{New Advances on GCRT}
\noindent In the following, we further improve the method in \cite{2016symmetric} to determine $\widetilde{q}_i$ defined in Subsection A. Here we assume that $r_{il}$ in each residue set are distinct or the repetition can be determined, i.e., $\kappa=NL$. The probability of such case holding is lower bounded by $\prod_{i=1}^{N} (1-\frac{(i-1)L}{\min \{\mathscr{M}\}})$, which is greatly high in practice. For the case that there exist repeated residues, the solution can be referred to \cite{2007sharpened}, which we will analyze later. We first recall our previous results about applying symmetric polynomials to develop GCRT. The well-known Viete Theorem tells the relationship between the coefficients of a polynomial and its roots,
\begin{lm} (Viete Theorem)
 Any polynomial of $N$-degree
\begin{equation}
P(x)=a_Nx^N+a_{N-1}x^{N-1}+\cdots +a_1x+a_0
\end{equation}
is known to have $N$ roots $\{x_1,x_2, \ldots,x_N\}$ by the fundamental theorem of algebra and relationships between the roots. The coefficients are:
\begin{equation}
\begin{cases}
\sum_{i=1}^{N} x_i=-\frac{a_{N-1}}{a_N}=c_1\\
\sum_{1\leq g<h \leq N}x_g x_h=\frac{a_{N-2}}{a_N}=c_2\\
~~~~~~~~~~~~~~~\cdots\\
\prod_{i=1}^{N}x_i=(-1)^N \frac{a_0}{a_N}=c_N
\end{cases}
\end{equation}
\end{lm}

The converse is also true. To recover $\widetilde{q}_i$, it is equivalent to finding a $N$-degree polynomials, of which the roots are $\widetilde{q}_i$. Calculating the coefficients $c_i$ can be implemented by computing the symmetric polynomial of $\widetilde{q}_i$ in each $\mathbb{Z}/M_l$ via (71). The residues of $c_i$ modulo $M_l$ is determined by the symmetric polynomials in (71) of the residues $\widetilde{q}_{il}$. Due to the symmetry, there is no need to distinguish the correspondence. However, it is a computation intensive task if the Viete Theorem is directly applied to obtain $c_i$, which require $O(NL2^N)$ modular multiplication (MM) operations. We further introduce the following scheme to reduce the complexity sharply.

\begin{lm} (Viete-Newton Theorem)
 Given $N$ integers $\{x_1,x_2, \ldots,x_N\}$ and power sum symmetric polynomials, $S_k=\sum_{i=1}^{N} x_i^k~, k=1,2, \ldots,N$, solving the equations $S_k=p_k, k=1,2, \ldots,N$, is equivalent to solving
\begin{equation}
	P(x)=x^N-c_1x^{N-1}+\cdots+(-1)^N c_N=0
\end{equation}
where $c_0=1$ and
\begin{equation}
c_i = \frac{1}{i}\sum_{j=1}^{i}(-1)^{j-1}p_jc_{i-j}, 1\leq i \leq N
\end{equation}
\end{lm}

Therefore, constructing $\{S_i\}, i=1,2,...,N$, only requires $(N-1)NL$ MM operations and totally constructing $\{c_i\}, i=1,2,...,N,$ requires $O(N^2L)$ MM operations. The complexity of finding roots of an integer coefficient polynomials has been proved in polynomial time with the famous LLL lattice reduction in \cite{LLL}, which takes $O(N^6+N^5(log {(\sum_{i=1}^{N} {c_i}^2   )}^{\frac{1}{2}}))$ arithmetic operations. A great deal of improved methods have been proposed during the last three decades, where we omit the details as it is out of the scope of discussion in this paper. In fact, the most exciting property of the constructed $P(x)$ is that it has $N$ integer roots $\widetilde{q}_i$.  Therefore for a prime $p > \max_{i} \widetilde{q}_i$, the roots of $P(x)$ over $\mathbb{F}_p$ are still $\widetilde{q}_i$. Under Extended Riemann Hypothesis, a factoring algorithm with running time ${O(N^{logN}log p)}^{O(1)}$ is proposed by Evdokimov \cite{ERH}. In real implementation, besides algebraic methods, we believe numerical solutions such as the Newton-Raphson method are more efficient as $P(x)$ only has integer roots and a very small probability that there exist repeated roots. After we introduce the fundamental theory above, the key issue left is how to make $c_i$ as small as possible given dynamic range $K$. Before we present the results, we first define a negative number in $\mathbb{Z}/m$.

\begin{lm} \label{negative} For an integer $X$ such that $|X|<\frac{m}{2}$, if $X<0$, then $\langle X \rangle_{m} \geq \frac{m}{2}$ and if $X \geq 0$, then $\langle X \rangle_{M}<\frac{m}{2}$.
\end{lm}
\indent Lemma \ref{negative} is used to check the sign of $X$.

\begin{rmk} \label{margin}
If $ X_i<\delta $ and $X_i+\Delta r_{il}<0$ for $ l=1,2,...,L $, the dynamic range should be doubled. Therefore, $\widetilde{q}_i=\langle \prod_{l=1}^{L} M_{l}-1 \rangle_{\prod_{l=1}^{L} M_{l}}=-1$ can be recovered and the signs of the coefficients in the symmetry polynomials are checked with Lemma \ref{negative}. Interested readers my refer to \cite{rnscomparison} for more details.
\end{rmk}

\indent Inspired by the ideas in \cite{2016symmetric}, we develop the following symmetry polynomials.
\begin{equation}
\begin{aligned}
      \Lambda=& \{ e_1=\sum_{i=1}^{N}\{\widetilde{q}_i-{\widetilde{q}}\}, e_2=\sum_{1\leq i<j \leq n}(\widetilde{q}_i-{\widetilde{q}})(\widetilde{q}_j-{\widetilde{q})},\\
                       & \ldots, {e_N=\prod_{i=1}^{N}{(\widetilde{q_i}-{\widetilde{q})}}} \}
\end{aligned}
\end{equation}
where $\widetilde{q}= \frac{\sum_{i=1}^{N}\widetilde{q}_i}{N} $. If the condition of $\prod_{l=1}^{L} M_{l} > \max\{\sum_{i=1}^{N} \widetilde{q}_{i},~2\binom{N}{i}d^{i},i=2,3,...,N\}$ is satisfied, $\{\widetilde{q}_i, i=1,2,...,N\}$ can be uniquely determined using $\{\widetilde{q}_{il}\}$, where $d=\max_i\{\widetilde{q}_i\}-\min_i\{\widetilde{q}_i\}$. Furthermore, if $d\geq N$, the condition above is equivalent to $\prod_{l=1}^{L} M_{l} > max\{\sum_{i=1}^{N} \widetilde{q}_{i}, 2d^{N} \}$. In fact the bound can be further reduced to be $max\{\sum_{i=1}^{N} \widetilde{q}_{i}, 2\binom{N}{i}({\frac{d}{2}})^{i},i=2,3,...,N\}$

 Based on Lemma \ref{negative}, we further extend the symmetric-polynomial GCRT. Considering the following symmetric polynomial set, we give Theorem \ref{symmbound}.

\begin{thm} \label{symmbound} For $N$ real numbers, $\{Z_i, i=1,2,...,N\}$, arranged in an ascending order such that $\sum_{i=1}^{N} Z_i=0$ and $0 \leq Z_N-Z_1\leq d$, then
\begin{equation}
|\sum_{1 \leq i_1 <...<i_V \leq N} \prod_{k=1}^{V} Z_{i_k}| \leq \binom{N}{V} (d/2)^{V}
\end{equation}
where  $V \in \{2,3,...,N\}$ and $\{i_k\}$ is a $V$-dimensional subset of $\{1,2,...,N\}$.
\end{thm}

\begin{proof} Assume set $\{Z_i\}$ is consisted of $S$ positive and $N-S$ non-positive elements. Let
\begin{equation}
\begin{aligned}
F(Z_1,Z_2,...,Z_N)=& \sum_{1 \leq i_1 <i_2<...<i_V \leq N} |\prod_{k=1}^{V} Z_{i_k}| \\
& \geq |\sum_{1 \leq i_1 <i_2<...,<i_V \leq N} \prod_{k=1}^{V} Z_{i_k}|
\end{aligned}
\end{equation}
\indent Let $\sum_{i=1}^{N-S}Z_i=-A$ and $\sum_{i=N-S+1}^{n}Z_i=A$, where $A$ is a positive number. Assume that $a$ elements of the symmetric sum in (76), i.e., $Z_{i_1},...,Z_{i_{a}}$, are selected from $\{Z_1,Z_2,...,Z_{N-S}\}$ and the other $V-a$ elements, i.e., $Z_{i_{a+1}},...,Z_{i_{V}}$, are selected from $\{Z_{N-S+1},Z_{N-S+2},...,Z_{N}\}$. Let $\Theta_1=\sum_{1 \leq i_1 <...<i_{a} \leq N-S} |\prod_{k=1}^{a} Z_{i_k}|$ and $\Theta_2=\sum_{N-S+1 \leq i_{a+1} <...<i_{V} \leq N} \prod_{k=a+1}^{v} |Z_{i_k}|$, thus (76) is rewritten as $F(Z_1,Z_2,...,Z_N)=\Theta_1 \times \Theta_2$.\\
\indent According to Muirhead's inequality, $\Theta_1=\sum_{1 \leq i_1 <i_2<...<i_{a} \leq N-S} \prod_{k=1}^{a} |Z_{i_k}|\leq \binom{N-S}{a} (\frac{\sum_{i=1}^{N-S}{Z_i}}{N-S})^{a}$. $\Theta_1$ achieves the maximum value iff $Z_1=Z_2=...=Z_{N-S}=-\frac{A}{N-S}$. Similarly, $\Theta_2$ achieves the maximum value iff $Z_{N-S+1}=Z_{N-S+2}=...=Z_{N}=\frac{A}{S}$. In the following, we introduce the method of adjusting the values of elements in $\{Z_i\}$ step-by-step to maximize $F(Z_1,Z_2,...,Z_N)$. Firstly, if $Z_N-Z_1 \neq d$, simultaneously increase $Z_N$ and decrease $Z_1$ until that $Z_N-Z_1=d$. Meanwhile, keep $F$ not to decrease. Furthermore, we can adjust the values of elements in $\{Z_i\}$ to satisfy that $\frac{A}{N-S}+\frac{A}{S}=d$. Therefore, all positive elements equal to $\frac{N-S}{N}d$ and all non-positive elements equal to $-\frac{S}{N}d$. Without loss of generality, assume that $S \leq N-S$ due to the symmetry property, which leads to $\frac{S}{N}d \leq \frac{d}{2}$. We can transform (76) to the following form,
\begin{equation}
F(Z_1,Z_2,...,Z_N)=(|Z_1|+|Z_N|)\Psi_1+(|Z_1Z_N|)\Psi_2+\Psi_3
\end{equation}
where $\Psi_1, \Psi_2$ and $\Psi_3$ are symmetric polynomials of the rest elements. Since $|Z_1Z_N| \leq \frac{(|Z_1|+|Z_N|)^2}{4} \leq \frac{d^2}{4}$, the equality holds iff $Z_1=-\frac{d}{2}$ and $Z_N=\frac{d}{2}$. Therefore, we have $F(Z_1,Z_2,...,Z_N)\leq F(-\frac{d}{2},Z_2,...,Z_{N-1},\frac{d}{2})$. The above operation is performed iteratively until there exist $Z_i$ and $Z_{N-i+1}$ that are both non-positive elements. In this case, because $Z_N-Z_1=d$, we have $|Z_i|+|Z_{N-i+1}|\leq d$ and $|Z_iZ_{N-i+1}|\leq \frac{d^2}{4}$. Hence we obtain
\begin{equation}
F(Z_1,Z_2,...,Z_N)\leq F(-\frac{d}{2},-\frac{d}{2},...,\frac{d}{2})
\end{equation}
\end{proof}

\indent Especially, if $d \geq 2N$, the condition is reduced to $\prod_{l=1}^{L} M_{l} > \max_{i} \{\sum_{i=1}^{N} \widetilde{q}_{i}, 2(\frac{d}{2})^{N}\}$. Based on Theorem \ref{symmbound} and Lemma \ref{negative}, a new GCRT is developed as Algorithm 2.\\
\begin{algorithm}[h]
	\caption{}
	\textbf{Input}  $\widetilde{q}_{il}$ for $i=1,2,...,N$ and $l=1,2,...,L$.  \\
Step-1. Calculate $\langle \sum_{i=1}^{N} \widetilde{q}_{il} \rangle_{M_l}$ for $l \in \{1,2,...,L\}$ and recover $\widetilde{q}=\frac{\widetilde{q}_{i}}{N}$. \\
Step-2. Calculate $\langle \widetilde{q} \rangle_{M_l}$. \\
Step-3. Calculate $\langle\{e_1,e_2,...,e_N\}\rangle_{M_l} $ for $l \in \{1,2,...,L\}$. \\
Step-4. Recover $e_1,e_2,...,e_N$ with CRT. \\
Step-5. Let $e_0=1$ and construct the polynomial $P(x)=\sum_{i=0}^{N}(-1)^ie_ix^{N-i}$.\\
Step-6. Solve the equation $P(x)=0$ and get $N$ roots $\{\widetilde{q}_1-\widetilde{q},\widetilde{q}_2-\widetilde{q},...,\widetilde{q}_N-\widetilde{q}\}$. Therefore, $\{\widetilde{q}_1,\widetilde{q}_2,...,\widetilde{q}_N\}$ are obtained.\\
        \textbf{Output}  $\widetilde{q}_1,\widetilde{q}_2,...,\widetilde{q}_N$
\end{algorithm}

Since CRT can be trivially generalized for real number residues \cite{2010closed}, $\widetilde{q}$ is not limited to be an integer. In the following part, we will discuss the lower bound of the dynamic range in order to recover multiple frequencies. Under the condition of no repeated residues existing, we have
\begin{equation}
\begin{aligned}
d &\leq \lfloor \frac{\max_i\{X_i\}+\delta}{\Gamma} \rfloor - \lfloor \frac{\min_i\{X_i\}-\delta}{\Gamma} \rfloor \\
&\leq \frac{\max_i\{X_i\}-\min_i\{X_i\}+2\delta}{\Gamma} +1.
\end{aligned}
\end{equation}
On the other hand,
\begin{equation}
\begin{aligned}
   \sum_{i=1}^{N}\widetilde{q}_i \leq \sum_{i=1}^{N}{q_i}+N < \lfloor \frac{ \sum_{i=1}^{N} X_i+\delta}{\Gamma} \rfloor +2N
   \end{aligned}
\end{equation}
Let $B=\max_i\{X_i\}-\min_i\{X_i\}$ denote the bandwidth in our model. According to Theorem 2, the lower bound is as follows.
\begin{equation}
\begin{aligned}
              \prod_{l=1}^{L}M_j>& \max_{k=2,3,...,N}\{\lfloor \frac{ \sum_{i=1}^{N} X_i+\delta}{\Gamma} \rfloor +2N, \\
              &2\binom{N}{k} (\frac{B+2\delta+\Gamma}{2\Gamma})^{k} \}
  \end{aligned}
\end{equation}
For the case that repeated residues exist, referring to the (\cite{2007sharpened}, Theorem 2), the lower bound is
\begin{equation}
       \prod_{l=1}^{\lceil \frac{L}{N} \rceil} {M_l}>max_i\{q_i\}+1=\lfloor \frac{ max_i\{X_i\}}{\Gamma} \rfloor+1
\end{equation}

\textbf{Example}. Assume that $m_1=50\times7=350, m_2=50\times9=450, m_3=50\times11=550, m_4=50\times13=650$ and $X_1=1110, X_2=1995, X_3=2016$. Thus $r_{11}=60, r_{12}=210, r_{13}=0, r_{14}=460$; $r_{21}=245, r_{22}=195, r_{23}=345, r_{24}=460$; $r_{31}=246, r_{32}=216, r_{33}=366, r_{34}=66$. Let $\delta=4<\frac{50}{4\times3}$ and assume that the erroneous residue sets are $\{\widetilde{r}_{11}=64, \widetilde{r}_{12}=206, \widetilde{r}_{13}=547, \widetilde{r}_{14}=462\}$, $\{\widetilde{r}_{21}=247, \widetilde{r}_{22}=192 \widetilde{r}_{23}=348, \widetilde{r}_{24}=48\}$ and $\{\widetilde{r}_{31}=250, \widetilde{r}_{32}=213, \widetilde{r}_{33}=370, \widetilde{r}_{34}=62\}$. Accordingly, we get \{$\widetilde{r}^{c}_{11}=14, \widetilde{r}^{c}_{12}=6, \widetilde{r}^{c}_{13}=47, \widetilde{r}^{c}_{14}=12$; $\widetilde{r}^{c}_{21}=47, \widetilde{r}^{c}_{22}=42 \widetilde{r}^{c}_{23}=48, \widetilde{r}^{c}_{24}=48$; $\widetilde{r}^{c}_{31}=0, \widetilde{r}^{c}_{32}=13,  \widetilde{r}^{c}_{33}=20, \widetilde{r}^{c}_{34}=12$\}. After $\widetilde{r}^{c}_{il}$ are arranged in ascending order, we find that $(\widetilde{r}^{c}_{22}=42)-(\widetilde{r}^{c}_{33}=20)=22>2\delta$. Therefore, let $\hat{r}^{c}_{il}=\widetilde{r}^{c}_{il}$ if $\widetilde{r}^{c}_{il} \leq 20$, otherwise $\hat{r}^{c}_{il}=\widetilde{r}^{c}_{il}-\Gamma$.

Then, $\widetilde{q}_{ij}= \lfloor \widetilde{r}_{ij}-\hat{r}^{c}_{ij} \rfloor$ are obtained, which are the residues of $\widetilde{q}_{i}$ modulo $M_l$, i.e., $\widetilde{q}_{11}=1, \widetilde{q}_{12}=4, \widetilde{q}_{13}=0, \widetilde{q}_{14}=9$; $\widetilde{r}_{21}=5, \widetilde{r}_{22}=4, \widetilde{r}_{23}=7, \widetilde{r}_{24}=1$; $\widetilde{r}_{31}=5, \widetilde{r}_{32}=4, \widetilde{r}_{33}=7, \widetilde{r}_{34}=7$. With Algorithm 1, we recover $ \sum_{i=1}^{3} \widetilde{q}_{i}=102 $ with the residues $\langle \sum_{i=1}^{3} \widetilde{q}_{il} \rangle_{M_l}$, i.e., $\{4,3,3,11\}$. Furthermore, $\widetilde{q}=\frac{102}{3}=34$, where $\{6,7,1,8\}$ are its residues modulo $M_1=7, M_2=9, M_3=11, M_4=13$. Then residue sets of $\{\widetilde{q}_{i}-\widetilde{q}, i=1,2,3\}$ modulo $\{M_l, l=1,2,3,4\}$ are $\{ -5,-1,-1\}$, $\{ -3,-3,-3\}$, $\{ -1,6,6\}$ and $\{ 1,-7,-7\}$, respectively. The residues of $\{e_1, e_2,e_3\}$ modulo $\{M_l, l=1,2,3,4\}$ are $(4,0,2,9)$, $(2,0,8,10)$, $(0,0,0)$, respectively. Therefore we construct the polynomial $P(x)=x^3-108x+432$ and solve it to get roots $\{-12,6,6 \}$, which indicates that $\widetilde{q}_{1}=22, \widetilde{q}_{2}=40, \widetilde{q}_{3}=40$ and the corresponding relation between $\widetilde{q}_{i}$ and $ \widetilde{r}^{c}_{il}$. Finally, the estimated frequencies are obtained as $\widetilde{X}_i=\widetilde{q}_{i}\Gamma + [\frac{\sum_{l=1}^{3} \hat{r}^{c}_{il}}{3}]$, i.e., $\widetilde{X}_1=1110$, $\widetilde{X}_2=1996$ and $\widetilde{X}_3=2016$.

As a final remark, the final construction has been divided into $N$ steps via (65) for each $\widetilde{X}_i$ seperately in the proposed GRCRT, which is the same as the final step for recovery in \cite{2010closed} in essence. Therefore, for each estimation $\widetilde{X}_i$ of $X_i$, the performance is the same as that of \cite{2010closed}. The overall performance of GRCRT can be regarded as that of $N$ independent reconstruction for a single integer in \cite{2010closed} when robustness, i.e., $\delta < \frac{\Gamma}{4N}$, can be achieved.

\section{Conclusion}
\noindent In the paper, we investigated the robustness in CRT from theory to applications. To address CRT-based frequency estimation from undersampling waveforms, we proposed the first RCRT for multiple integers as a complete theoretical solution. Further improvement including a weakening bound of moduli leveraging symmetric polynomials is  developed. Besides, we analyzed the robustness in conventional RCRT from a geometry perspective and proposed the shift pseudo metric to present a general framework to study trade off between dynamic range $K$ and error bound $\delta$ for such Lee-metric based remainder code. Thanks to all the previous works, we step forward a complete solution for the problem.



\end{document}